\documentclass[reqno]{amsart}

\def\baselinestretch{1}

\usepackage{multicol}
\usepackage{amsmath}
\usepackage{graphicx}
\usepackage[misc,geometry]{ifsym}
\usepackage{hyperref}
\usepackage{amssymb,latexsym,amsfonts,amsmath}
\usepackage{caption,subcaption}
\usepackage{graphicx}
\usepackage{paralist}
\usepackage{dsfont}
\usepackage{booktabs}
\usepackage{eso-pic}
\usepackage{epsfig}
\usepackage{url}
\usepackage{epstopdf}
\usepackage{tikz}
\usetikzlibrary{automata,positioning}
\usepackage{bm}
\usepackage{verbatim}

\usepackage{color}
\usepackage{mathtools}

\usepackage{etoolbox}
\AtBeginEnvironment{thebibliography}{\linespread{1}\selectfont}

\usepackage{cite}

\usepackage{mdframed,lipsum}

\usepackage{mathrsfs}

\usepackage{algorithm}
\usepackage{algorithmicx}
\usepackage[noend]{algpseudocode}
\makeatletter
\newcommand{\ALOOP}[1]{\ALC@it\algorithmicloop\ #1%
  \begin{ALC@loop}}
\newcommand{\ENDALOOP}{\end{ALC@loop}\ALC@it\algorithmicendloop}

\makeatother\algnewcommand{\Initialize}[1]{%
	\State \textbf{Initialize:}
	\Statex \hspace*{\algorithmicindent}\parbox[t]{.8\linewidth}{\raggedright #1}
}
\algnewcommand\And{\textbf{and} }

\usepackage{url}
\usepackage{graphicx}
\usepackage{caption}

\newcommand{\R}{{\mathbb{R}}}

\newcommand{\B}{\mathrm{B}}
\newcommand{\V}{\mathcal{V}}
\newcommand{\E}{\mathcal{E}}
\newcommand{\G}{\mathcal{G}}

\newcommand{\N}{{\mathbb{N}}}

\newcommand{\bfm}{\mathbf}

\newcommand{\diag}{\mathrm{diag}}
\newcommand{\xsig}{\bfm{x}_{x_0,\bm{\sigma}}}
\newcommand{\np}{\mathrm{np}}
\newcommand{\safe}{\mathrm{Safe}}
\newcommand{\qedsymb}{\hfill\blacksquare}

\newtheorem{theorem}{Theorem}[section]
\newtheorem{lemma}[theorem]{Lemma}

\newtheorem{problem}[theorem]{Problem}
\newtheorem{proposition}[theorem]{Proposition}

\newtheorem{definition}[theorem]{Definition}
\newtheorem{example}[theorem]{Example}

\newtheorem{remark}[theorem]{Remark}

\numberwithin{equation}{section}

 \usepackage{enumerate}
\usepackage[shortlabels]{enumitem}

\usepackage[many]{tcolorbox}
\usetikzlibrary{calc}
\tcbuselibrary{skins}

\newtcolorbox{resp}[1][]{%
	enhanced jigsaw,%
	colback=gray!5!white,%
	colframe=gray!80!black,%
	size=small,%
	boxrule=1pt,%
	halign title=flush center,%
	coltitle=black,%
	breakable,%
	drop shadow=black!50!white,%
	attach boxed title to top left={xshift=1cm,yshift=-\tcboxedtitleheight/2,yshifttext=-\tcboxedtitleheight/2},%
	minipage boxed title=3cm,%
	boxed title style={%
		colback=white,%
		size=fbox,%
		boxrule=1pt,%
		boxsep=2pt,%
		underlay={%
			\coordinate (dotA) at ($(interior.west) + (-0.5pt,0)$);
			\coordinate (dotB) at ($(interior.east) + (0.5pt,0)$);
			\begin{scope}[gray!80!black]
				\fill (dotA) circle (2pt);
				\fill (dotB) circle (2pt);
			\end{scope}
		}%
	},%
	#1%
}

\renewcommand{\emptyset}{{\varnothing}}

\title{On the Completeness and Ordering of Path-Complete Barrier Functions} 

\author{Mahathi Anand$^1$}
\author{Rapha\"el Jungers$^2$}
\author{Majid Zamani$^3$}
\author{Frank Allg\"ower$^4$}
\address{$^1$Munich Institute of Robotics and Machine Intelligence (MIRMI),
Technical University of Munich (TUM), Germany}
\email{mahathi.anand@tum.de}
\address{$^2$Department of Mathematical Engineering, Universit\'e catholique de Louvain.}
\email{raphael.jungers@uclouvain.be}
\address{$^3$University of Colorado Boulder, USA}
\email{majid.zamani@colorado.edu}
\address{$^4$Institute for Systems Theory and Automatic Control, University of Stuttgart, Germany.}
\email{frank.allgower@ist.uni-stuttgart.de}

\begin{document}

\begin{abstract}                          
This paper is concerned with path-complete barrier functions which offer a graph-based methodology for verifying safety properties in switched systems. The path-complete framework leverages algebraic (barrier functions) as well as combinatorial (graph) components to characterize a set of safety conditions for switched systems, thus offering high flexibility (two degrees of freedom) in searching for suitable safety certificates. In this paper, we do not propose any new safety criteria. Instead, we further investigate the role that the combinatorial component plays in the safety verification problem. First, we prove that \emph{path-completeness}, which is a property on a graph that describes the switching sequences, is \emph{necessary} to obtain a set of valid safety conditions. As a result, the path-complete framework is able to provide a complete characterization of safety conditions for switched systems. Furthermore, we provide a systematic methodology for comparing two path-complete graphs and the conservatism associated with the resulting safety conditions. Specifically, we prove that under some conditions, such as when there exists a simulation relation between two path-complete graphs, it is possible to conclude that one graph is always able to provide less conservative safety conditions than another, independent of the algebraic properties of the switched system and the template of the barrier function under consideration. Such a result paves the way for a systematic use of the path-complete framework with barrier functions, as one can then consistently choose the appropriate graph that provides less conservative safety conditions. 
\end{abstract}

\maketitle
\thispagestyle{empty}
\pagestyle{empty}

\maketitle
\thispagestyle{empty}
\pagestyle{empty}

\section{Introduction}

Traditionally, control theorists have focused on studying systems with purely continuous or discrete behaviors. However, many physical systems behave in a hybrid manner. Examples include chemical processes and energy systems~\cite{hybrid_chem}, robotics~\cite{hybrid_robotics}, biological systems~\cite{hybrid_bio}, and air traffic management systems~\cite{hybrid_traffic}. On the other hand, switched systems~\cite{Liberzon2003SwitchingIS} exhibit continuous-time behavior with isolated discrete switching modes and can be used as abstractions for more complex hybrid systems. However, their analysis is still a challenging problem~\cite{switched_opt,switched_stab} that requires careful consideration.  

An important problem concerning switched systems is the analysis of safety properties. The safety verification problem entails ascertaining whether all the behaviors of the switched system avoid visiting some unsafe or undesirable configurations. One way to solve this problem is to construct simpler finite abstractions and analyze the abstract systems, e.g., in~\cite{switched_abs1,switched_abs2,switched_abs3}, however, they suffer severely from the curse of dimensionality. More recently, discretization-free methods using barrier functions~\cite{prajna_safety_2004,prajna_barrier_2006} have shown great promises in proving safety properties effectively and efficiently. They are real-valued functions that assign higher values to the unsafe states than the initial states, thus serving as a \emph{barrier} between reachable and unsafe regions. Hence, the safety verification problem is reduced to finding suitable barrier functions. 

\textbf{Related Literature.} Several results in the literature utilize barrier functions in the context of switched and hybrid systems~\cite{prajna_safety_2004,cbf_switch1,cbf_switch2,asymp_stab_safe,exp_barr,state-time-safety,anand_switched,stoc_swit,stoc_hyb,nejati_comp_stoc}. The existing approaches can be divided in two groups. The first one is the use of common barrier functions~\cite {cbf_switch2,cbf_switch1}, where a single barrier function is used to guarantee safety for all the discrete operating modes of the switched system. Unfortunately, the conditions imposed on the barrier function are quite restrictive, and one may not be able to find suitable common barrier functions even when the system is safe. To alleviate the conservatism, multiple barrier functions were proposed in~\cite{prajna_safety_2004,asymp_stab_safe,exp_barr,state-time-safety}. Here, a barrier function is assigned to each operating mode of the switched (hybrid) system, and safety guarantees are obtained by ensuring not only mode safety but also safety at switching instances. 

More recently, inspired by path-complete Lyapunov functions for stability analysis~\cite{pclf_original,pclf_common,geom_pclf,pclf_ordering1,pclf_ordering2,temp_lifts_pclf}, we introduced in~\cite{anand_pcbf} a new framework for the safety of \emph{linear} switched systems under arbitrary switching, based on a notion of \emph{path-complete barrier functions}. This approach extends and generalizes the common and multiple barrier functions framework by introducing a combinatorial component to the safety verification problem via a graph. In particular, a so-called \emph{path-complete graph} is used to encode the arbitrary switching sequences occurring in the system. Then, path-complete barrier functions are defined as a collection of functions so that each node of the graph is assigned a distinct barrier function. The edges within the graph enable algebraic conditions imposed on the barrier functions, such that the structure of the graph, together with the algebraic conditions, provide sufficient conditions for safety. It was also shown in~\cite{anand_pcbf} that depending on the choice of path-complete graphs, the approach provides less conservative results compared to the existing common and multiple barrier function framework. However, some unresolved questions regarding the path-complete framework remain. In this paper, we investigate in detail the special role that the choice of graphs plays in solving the safety verification problem. In this context, we address important questions concerning the completeness and ordering of path-complete barrier functions.  

\textbf{Contributions.} The contributions of this paper in comparison to~\cite{anand_pcbf} are three-fold. First, while~\cite{anand_pcbf} focused primarily on linear switched systems, we remove in this paper any such restrictions and provide results for general nonlinear systems. Second, we provide completeness arguments for the path-complete framework proposed in~\cite{anand_pcbf}. It is known that path-complete barrier functions only provide sufficient conditions for safety. In other words, given a system and a suitable path-complete graph, failure to find a barrier function does not mean that the system is unsafe. However, it is not yet clear whether the chosen graph is required to be path-complete to begin with. Therefore, in this paper, we ask whether path-completeness is necessary for proving safety properties, and answer this question affirmatively by constructing appropriate counterexamples, i.e., by constructing suitable unsafe systems where the lack of the path-completeness property in graphs may lead to invalid safety certificates. This enables us to prove that the path-complete framework encapsulates the complete set of graph-based conditions for the safety verification of switched systems.

Thirdly, we investigate how the choice of graphs can influence the conservatism of the resulting safety conditions. In general, for a given system, the choice of a path-complete graph over which safety conditions are defined is not unique. As a result, different path-complete graphs lead to different algebraic conditions for safety verification, and the conservatism of the resulting conditions may differ. In this paper, we leverage techniques from graph theory to compare two path-complete graphs in terms of the conservatism of the resulting safety conditions. In particular, we define some simulation relations on the graphs under which one can order two path-complete graphs according to their conservatism, independently of the system or the barrier function template under consideration. Then, for any two path-complete graphs with some simulation relation, it is possible to objectively conclude in a complete manner, which graph results in less conservative safety certificates. This allows to consistently choose an appropriate graph that provides less conservative safety conditions, thus bridging the existing gap in the path-complete framework and making the approach more useful for verifying safety properties for switched systems systematically. 

In addition, we support our theoretical results with several illustrative examples and experiments. First, we utilize a nonlinear two-car platoon system to demonstrate the applicability of path-complete barrier functions for the safety of switched nonlinear systems. Then, to illustrate the completeness of the path-complete barrier function framework, we construct counterexamples, both linear and nonlinear, to demonstrate that unsafe systems can yield invalid safety guarantees when graphs are chosen to be non path-complete. Finally, we conduct experiments to investigate the consequences of the choice of graphs in the resulting safety certificates. We show that in the presence of a simulation relation between two path-complete graphs, one graph always performs better than the other. 

We note that our results echo similar existing ones in the context of path-complete Lyapunov functions presented in~\cite{jungers_characterization_2017,geom_pclf,pclf_ordering1}. Particularly, in~\cite{jungers_characterization_2017}, the necessity of path-completeness for stability criteria were provided, and~\cite{geom_pclf, pclf_ordering1} tackled the comparison between path-complete Lyapunov functions. Though our results are analogous to these, we point out that there is no reason (to the best of our knowledge) to assume that these results can be extended directly from the stability to the safety framework. Indeed, it turns out that the proofs are significantly different, even though they are based on similar intuition. This is due to the fact that the algebraic conditions for stability and safety are also quite different.
 
\textbf{Outline.} In Section~\ref{sec:prelim}, we first highlight the necessary notations and then present the required background on switched systems as well as the safety verification problem. Consequently, in Section~\ref{sec:pcbf}, we summarize briefly the path-complete framework for the safety verification of switched systems that was first proposed in~\cite{anand_pcbf}, and extend the framework to general nonlinear systems. Section~\ref{sec:nec} presents the first main result of this paper. Here, we address the completeness of the path-complete barrier function framework and construct counterexamples which show that path-completeness of graphs is indeed a necessary condition for safety verification. Section~\ref{sec:comparison} then addresses the ordering of path-complete barrier functions. Here, necessary and sufficient conditions on any two graphs are presented in order to compare the conservatism of the safety conditions resulting from the graphs. In Section~\ref{sec:cs}, we present some case studies and experiments to validate the theoretical results obtained in the paper. Finally, we summarize our results and raise a few open questions in Section~\ref{sec:conclusion}.

\section{Notations and Problem Statement}
\label{sec:prelim}

\subsection{Notations}
\label{subsec:not}

Throughout the article, the set of real and non-negative integers are denoted by $\R$ and $\N$, respectively. In addition, appropriate subscripts are used to refer to subsets of $\R$, and $\N$, respectively, e.g., $\N_{\geq 1}$ for representing the set of positive integers. The notation $\R^n$ denotes a real space of dimension $n$ and $\R^{m \times n}$ denotes a real space of dimension $m \times n$. For a vector (denoted in lowercase) $x \in \R^n$, $x[i]$ denotes the $i^{\text{th}}$ element of $x$, $1 \leq i \leq n$. Similarly, for a matrix (denoted in uppercase) $X \in \R^{m \times n}$, the $(i,j)^{\text{th}}$ element is represented by $X[i,j]$, where $1 \leq i,j \leq n$. For a square matrix $X \in \R^{2n \times 2n}$, we write the $i^{\text{th}}$ $2 \times 2$ diagonal block of $X$ as $X[i] = \begin{bmatrix}
    X[2i-1, 2i-1] & X[2i-1, 2i] \\
    X[2i, 2i-1] & X[2i, 2i] \\
\end{bmatrix}$.
Given a vector $x \in \R^n$, $\diag(x) \in \R^{n \times n}$ is the diagonal matrix with diagonal elements $X[i,i] = x[i]$, $1 \leq i \leq n$. The inequality $A \geq 0 $ (resp $\leq$) is element-wise. 
Moreover, $A^\top$ denotes the transpose of $A$. The set of all continuous functions with domain $\R^n$ and co-domain $\R$ are denoted by $\mathcal{C}(\R^n,\R)$. 

Finally, an alphabet $\Sigma$ is a finite set of letters, and a sequence (denoted by bold symbols) $\bfm{w} = (w_0 w_1 \ldots) \in \Sigma^\omega$ is an infinite concatenation of letters, i.e., $w_i \in \Sigma$, for all $i \geq 0$. Similarly, a finite sequence of length $k$ is a sequence $\bfm{w}_k = (w_0 w_1 \ldots w_{k-1}) \in \Sigma^k$. We use $\bfm{w}(i) = w_i$ to denote the $i^{\text{th}}$ position of the sequence $\bfm{w}$. 

\subsection{Problem Definition}
\label{subsec:prob}

In this work, we consider nonlinear discrete-time switched dynamical systems (dt-SwS) of the form
\begin{equation} \label{eq:sysdyn}
  S := \bfm x(t+1)=f_{\bm \sigma(t)}(\bfm x(t)),
\end{equation}
where $\bfm{x}(t) \in X \subseteq \R^n$ and $\bm{\sigma}(t) \in \Sigma = \{1,\ldots,m\}$, $m \in \N_{\geq 1}$, denote the state of the system as well the switching mode at time $t \in \N$, such that each mode corresponds to the dynamics $f_{\bm \sigma(t)} \in \{f_1, \ldots, f_m\}$. We call $\bm{\sigma} = (\sigma_0 \sigma_1 \ldots)$ a switching sequence of the system $S$, and we denote by $\bfm{x}_{x_0,\bm{\sigma}} = (x_0 \bfm x(1) \bfm x(2) \ldots)$ the infinite sequence generated from the initial condition $\bfm{x}(0) = x_0$ under the switching sequence $\bm{\sigma}$.

This article deals with the safety verification problem for dt-SwS under arbitrary switching sequences, i.e., to ensure that the dt-SwS $S$ starting from a given initial set does not reach an unsafe set for any possible arbitrary switching sequence $\bm{\sigma}$. We present the formal definition of safety as follows. 
\begin{definition}[System Safety] \label{def:safety}
Given an initial set $X_0 \subseteq X$ and an unsafe set $X_u \subseteq X$ such that $X_0 \cap X_u \neq \emptyset$, a dt-SwS $S$ as in~\eqref{eq:sysdyn} defined over $\Sigma$ is said to satisfy the safety specification denoted by $\safe(X_0, X_u)$ if for any $x_0 \in X_0$ and any $\bm{\sigma} \in \Sigma^\omega$, one has that $\xsig(t) \notin X_u$, $\forall t \in \N$.  
\end{definition}
The safety verification problem can then be formulated as follows. 
\begin{problem}[Safety Verification] \label{prob:safety}
Given a dt-SwS as in~\eqref{eq:sysdyn} over $\Sigma$, an initial set $X_0$, and an unsafe set $X_u$, verify that $S$ satisfies $\safe(X_0,X_u)$.
\end{problem}

In our previous work~\cite{anand_pcbf}, we addressed Problem~\ref{prob:safety} by proposing a new barrier function-based approach, called path-complete barrier functions. This approach leverages algebraic (barrier functions) as well as combinatorial (graph) components to characterize a set of safety conditions for dt-SwS. In the following section, for the sake of completeness, we briefly review the path-complete framework via suitable descriptions of graph-based barrier functions. 

\section{Path-Complete Framework for Safety Verification}
\label{sec:pcbf}

In this section, we discuss the combinatorial and algebraic components that 
constitute the path-complete framework for the safety verification of dt-SwS. 
The combinatorial component is a graph that allows encoding the switching sequences occurring in the dt-SwS. In particular, a graph is a pair $\G = (\V, \E)$ defined on the alphabet $\Sigma$, where $\V$ is the set of vertices with cardinality $|\V|$, and $\E$ is the set of edges $(v, \sigma, v') \in \E$, with $v,v' \in \V$, and $\sigma \in \Sigma$. 
We say that a sequence $\bm{\sigma} = (\sigma_0 \sigma_1 \ldots) \in \Sigma^\omega$ is accepted by $\G$ if one can match this sequence to a path in $\G$, i.e., there exists corresponding sequence of edges $\bm{\pi} = \big( (v_0, \sigma_0, v_1) (v_1, \sigma_1, v_2), \ldots \big)$ such that $(v_i, \sigma_i, v_{i+1}) \in \E$, for all $i \in \N$. Now, one may also consider a finite sequence $\bm{\sigma}_k = (\sigma_0\ldots\sigma_{k-1})$ of length $k$ and extend the definition of acceptance to finite paths. Then, it is not difficult to see that if there exists a finite sequence $\bm{\sigma}_k$ that cannot be accepted by $\G$, then any infinite sequence $\bm{\sigma}$ containing $\bm{\sigma_k}$ cannot be accepted by $\G$. We now present the definition of path-complete graphs that encode the arbitrary switching sequences possible in a dt-SwS $S$.   

\begin{definition}[Path-Complete Graphs] \label{def:pcgraph}
A graph $\G=(\V,\E)$ over the alphabet $\Sigma$ is said to be path-complete (PC) if any finite sequence $\bm{\sigma _k}$ of any length $k \in \N_{\geq 1}$ is accepted by $\G$. 
\end{definition}

\begin{figure}
  \centering
  \includegraphics[width=0.4\linewidth]{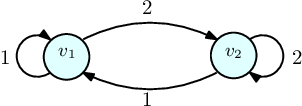}
  \caption{Example for a path-complete graph for a dt-SwS $S$ with $\Sigma = \{1,2\}.$}
  \label{fig:pcgraph}
\end{figure}

\begin{figure}
  \centering
  \includegraphics[width=0.4\linewidth]{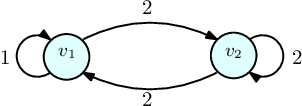}
  \caption{Example for a non path-complete graph for a dt-SwS $S$ with $\Sigma = \{1,2\}$. The sequence $\bm{\sigma}_3 = (121)$ cannot be accepted.}
  \label{fig:nonpc}
\end{figure}
An example of a path-complete graph for $\Sigma = \{1,2\}$ is shown in Figure~\ref{fig:pcgraph}. On the other hand, Figure~\ref{fig:nonpc} shows a graph that is not path-complete since it cannot accept any switching sequence $\bm{\sigma}$ containing the finite sequence $\bm{\sigma}_3 = (121)$. 

\begin{remark} \label{rem:pcchoice}
Note that given an alphabet $\Sigma$, the choice of a path-complete graph is in general non-unique. For example, reversing the direction of the edges $(v_1, 2, v_2)$ and $(v_2, 1, v_1)$ in Figure~\ref{fig:pcgraph} results in a new path-complete graph.
\end{remark}

The algebraic component of our framework is a set of real-valued functions that allows to characterize a set of inequalities via graphs. The following definition formalizes the notion of graph-based barrier functions, and as an extension, path-complete barrier functions, which are vital for addressing the safety verification problem (Problem~\ref{prob:safety}). 
\begin{definition}[Graph-Based Barrier Function] \label{def:pcbf}
Consider a dt-SwS $S$ over $\Sigma$ as in Definition~\ref{eq:sysdyn}, an initial set $X_0$, and an unsafe set $X_u$. A graph-based barrier function is a pair $\big (\G, \{\B_v, v \in \V\} \big)$, where $\G=(\V,\E)$ and $\B_v : X \rightarrow \R$ such that the following conditions are satisfied:
\begin{align}
  & \B_v(x) \leq 0, \quad && \forall x \in X_0, \ \forall v \in \V, \label{eq:pcbf1} \\
  & \B_v(x) > 0, \quad && \forall x \in X_u, \ \forall v \in \V, \label{eq:pcbf2} \\
  & \B_{v'}(f_{\sigma}(x)) \leq \B_v(x), && \forall x \in X, \ \forall (v, \sigma, v') \in \E. \label{eq:pcbf3}
\end{align}
In addition, $\big (\G, \{\B_v, v \in \V\} \big)$ is said to be admissible for $\G$ and $S$ if conditions~\eqref{eq:pcbf1}-\eqref{eq:pcbf3} are satisfied by $\{\B_v, v \in \V\}$.
\end{definition}
Moreover, a graph-based barrier function $\big (\G, \{\B_v, v \in \V\} \big)$ is said to be a \emph{path-complete barrier function} if the associated graph $\G$ is path-complete. In the following, we show that path-complete barrier functions can guarantee the satisfaction of safety specifications. 
\begin{theorem}
Consider a dt-SwS $S$ as in~\eqref{eq:sysdyn} over $\Sigma$, an initial set $X_0$ and unsafe set $X_u$, respectively. Suppose there exists a path-complete barrier function $\big (\G, \{\B_v, v \in \V\} \big)$ according to Definition~\ref{def:pcbf} for some path-complete graph $\G$. Then, $S$ satisfies $\safe(X_0,X_u)$. 
\end{theorem}
\begin{proof}
The proof is established by contradiction. Assume that $S$ is not safe, i.e., $\exists t \in \N$ such that $\xsig(t) \in X_u$ for some $x_0 \in X_0$ and a sequence $\bm \sigma = (\sigma_0\sigma_1\ldots)$. Suppose that a PCBF $\big (\G, \{\B_v, v \in \V\} \big)$ also exists for some path-complete graph $\G$ such that the switching sequence $\bm \sigma$ up to length $t$ generates the path $\bm{\pi}_t = \big((v(0)\sigma(0)v(1)) \ldots (v(t-1),\sigma(t-1),v(t))\big)$, where $v(i) \in \V$, $0 \leq i \leq t-1$ is the node reached in $\G$ at time $t$. Then from condition~\eqref{eq:pcbf1}, $\B_{v(0)}(x_0) \leq 0$. Moreover, due to Definition~\ref{def:pcgraph} and from condition~\eqref{eq:pcbf3}, we can generate the sequence $\bm \sigma$ for which we get $\B_{v(t)}(x(t))=\B_{v(t)}(A_{\sigma(t-1)}x(t-1)) \leq \B_{v(t-1)}(A_{\sigma(t-2)}x(t-2)) \leq \ldots \leq \B_{v(0)}(x_0) \leq 0$. This is in contradiction with condition~\eqref{eq:pcbf2}, which requires that for all $v(t) \in \V$, $\B_{v(t)}(x(t)) > 0$. So the system must satisfy $\safe(X_0,X_u)$. $\qedsymb$
\end{proof}

The existence of path-complete barrier functions is in general only sufficient for safety, but not necessary. This means that given a dt-SwS $S$ and a path-complete graph $\G$ over the alphabet $\Sigma$, failure to find a path-complete barrier function does not necessarily mean that $S$ is unsafe. However, it is not yet clear whether the requirement for the graph $\G$ to be path-complete is a necessary condition for safety. 
In other words, we ask whether the existence of a graph-based barrier function as in Definition~\ref{def:pcbf} can guarantee safety satisfaction for $S$ even if the associated graph is not path-complete. Answering this question allows us to provide a complete set of graph-based safety criteria for switched systems via our proposed path-complete framework. 

\section{Necessity of Path-Completeness}
\label{sec:nec}
In this section, we show that for any non path-complete graph $\G_{np}$, there exists a dt-SwS $S$ and a safety specification
$\safe(X_0,X_u)$ such that there exists an admissible graph-based barrier function for $\G_{np}$ and $S$ even when $S$ is unsafe. That is, $\G_{np}$  cannot lead to a valid safety certificate. We now formally present the main result of the section.

\begin{theorem}[Necessity of Path-Completeness] \label{thm:nec}
For any non path-complete graph $\G_{np}=(\V_{np}, \E_{np})$ over $\Sigma$, there exists a dt-SwS $S_{np}$ as in~\eqref{eq:sysdyn} and a safety specification $\safe(X_0,X_u)$ such that the following holds:
\begin{enumerate}
\item $S_{np}$ violates $\safe(X_0,X_u)$, \label{ite:1nec}
\item There exists an admissible graph-based barrier function $\big (\G_{np}, \{\B_v, v \in \V_{np}\} \big)$. \label{ite:2nec}
\end{enumerate}
\end{theorem}

The proof of Theorem~\ref{thm:nec} can be established by suitably constructing a dt-SwS $S_{np}$ and a specification $\safe(X_0,X_u)$ by leveraging the structure of $\G_{np}$. In particular, since $\G_{np}$ is not path-complete, there exists a sequence $\bm{\sigma}_k = (\sigma_0\ldots\sigma_{k-1})$ of length $k$ that cannot be accepted by $\G_{np}$. Using this, we construct $S_{np}$ and sets $X_0, X_u$ such that $S_{np}$ violates $\safe(X_0,X_u)$ upon encountering the switching sequence $\bm{\sigma} = (\bm{\sigma}_k \sigma_{k} \sigma_{k+1} \ldots)$, where $\bm{\sigma}$ is the sequence obtained by concatenating $\bm{\sigma}_k$ with an arbitrary sequence $(\sigma_{k} \sigma_{k+1} \ldots)$. In the following, we first provide a step-wise approach to the construction, and then conclude the section by summarizing the final proof of Theorem~\ref{thm:nec}. 

\subsection{Construction of System and Specification}

\label{subsec:cons}
Similar to the approach presented in~\cite{jungers_characterization_2017} in the context of stability, we start our construction with a non path-complete graph $\G_{np}=(\V_{np},\E_{np})$ over $\Sigma$ and identify a sequence $\bm{\sigma}_k =(\sigma_0\ldots\sigma_{k-1})$ of length $k$ that cannot be accepted by $\G_{np}$. Having this, we construct a dt-SwS $S_{np}$ as in~\eqref{eq:sysdyn} with linear dynamics, i.e., $f_{\bm{\sigma}(t)} \in \{A_{\sigma_1}, \ldots, A_{\sigma_m}\}$ for all $t \in \N, \forall \sigma \in \Sigma$,  and a corresponding safety specification $\safe(X_0,X_u)$ as follows.  

\begin{definition}[dt-SwS Construction] \label{def:cons}
For a non path-complete graph $\G_{np}$, consider a finite sequence $\bm{\sigma}_k \in \Sigma^k$, of length $k \in \N_{\geq 1}$ that is not accepted by $\G_{np}$. Define binary matrices $\hat A_\sigma \in \R^{n \times n}$, where $n=k+1$, such that $\hat A_\sigma[i,j] = 1$ if and only if $\bm{\sigma}_k(j-1) = \sigma$ and $i = j+1$, and all other elements of $\hat A_\sigma$ are identically zero. Then, define $A_\sigma = 2\hat A_\sigma$, $\forall \sigma \in \Sigma$. The linear dt-SwS $S_{np}$ may be obtained as
\begin{equation} \label{eq:unsys}
S_{np} \coloneq \bfm{x}(t+1) = A_{\sigma(t)} \bfm x(t),
\end{equation}
where $\bfm x(t) \in \R^n$, and $\bm{\sigma}(t) \in \Sigma, \forall t \in \N$. Moreover, define safety specification as $\safe(X_0,X_u)$ with 
\begin{align}
X_0 \coloneq  & \{ x \in \R^n \mid \sum_{i=1}^{n} x[i]^2 \leq \frac{1}{4^{k}} \}, \label{eq:init} \\
X_u \coloneq  & \{x \in \R^n \mid \sum_{i=1}^n x[i]^2 \geq 1 \}. \label{eq:uns}
\end{align}
\end{definition}

Clearly, by the construction in Definition~\ref{def:cons}, statement (\ref{ite:1nec}) of Theorem~\ref{thm:nec} is true, since for the initial condition $x_0 = [\frac{1}{2^k} \  0 \ \ldots \ 0]^\top \in X_0$ and the extended switching sequence $\bm{\sigma} = (\bm{\sigma}_k \sigma_{k} \sigma_{k+1} \ldots)$, we have that $\xsig(k) = [0 \ 0 \ \ldots \ 1]^\top \in X_u$, and as a result, $\safe(X_0,X_u)$ is violated. We now state the following lemma which is later used for proving statement (\ref{ite:2nec}) of Theorem~\ref{thm:nec}.

\begin{lemma} \label{lem:cons}
Let $\bm{\sigma}_k$ be a finite sequence over $\Sigma^k$, $k \in \N_{\geq 1}$. Consider the corresponding linear dt-SwS $S_{np}$ as in~\eqref{eq:unsys} as well as the safety specification $\safe(X_0,X_u)$ as in~\eqref{eq:init}-\eqref{eq:uns}. Consider a product of matrices from $\Sigma$ obtained as $A_{\bm{\sigma'}_l} = A_{\sigma_{l-1}'} \cdots A_{\sigma_{0}'}$, corresponding to a finite switching sequence $\bm{\sigma'}_l$, $l \in \N_{\geq 1}$. Then, $A_{\bm{\sigma'}_l} \neq 0$ if and only if $\bm{\sigma'}_l \in \bm{\sigma}_k$, i.e., there exist sequences $\bm{\sigma_m}, \bm{\sigma_n}$ of appropriate lengths such that $\bm{\sigma}_k = (\bm{\sigma}_m \bm{\sigma}'_l \bm{\sigma}_n$).
\end{lemma}
\begin{proof}
 This statement can be directly verified from the construction of $S_{np}$ via Definition~\ref{def:cons}.
\end{proof}

Now, to show the validity of Statement~\ref{ite:2nec} of Theorem~\ref{thm:nec}, we seek to construct a suitable graph-based barrier function $\big (\G_{np}, \{\B_v, v \in \V_{np}\} \big)$ for $S_{np}$ that satisfies conditions~\eqref{eq:pcbf1}-\eqref{eq:pcbf3}. For this purpose, we consider our barrier functions to be quadratic, so that for each $v \in \V_{np}$, we have $B_v : \R^n \rightarrow \R$ as 
\begin{equation} \label{eq:gbfnec}
\B_v(x) \coloneq \sum_{i=1}^n p_v[i]x[i]^2 - 1.
\end{equation}
Moreover, we let $p_v[i]$ be any real value satisfying
\begin{equation} ~\label{eq:sat}
1 < p_v[i] \leq 4^k,
\end{equation}
for all $v \in \V_{np}$ and $1 \leq i \leq n$. Note that condition~\eqref{eq:sat} is required for the satisfaction of~\eqref{eq:pcbf1}-\eqref{eq:pcbf2}, as we will see later. 

Now, it remains to show that the graph-based barrier function $\big (\G_{np}, \{\B_v, v \in \V_{np}\} \big)$ satisfies conditions~\eqref{eq:pcbf1}-\eqref{eq:pcbf3} in order to prove statement (\ref{ite:2nec}) of Theorem~\ref{thm:nec}. 

\subsection{Admissibility of Graph-Based Barrier Function}

In this section, we show that there exists a graph-based barrier function with the template~\eqref{eq:gbfnec} satisfying conditions~\eqref{eq:pcbf1}-\eqref{eq:pcbf3} for $S_{np}$ with respect to the corresponding safety specification $\safe(X_0,X_u)$ as in Definition~\ref{def:cons}. In particular, we first show that for $S_{np}$, condition~\eqref{eq:pcbf3} can be represented as simple vector inequalities over the coefficients $p_v$ in equation~\eqref{eq:gbfnec}. Then, we utilize some interesting properties of the graph $\G_{np}$ to find the coefficients $p_v$ that lead to the satisfaction of inequalities~\eqref{eq:pcbf1}-\eqref{eq:pcbf3}.

\begin{proposition} \label{prop:vectorineq}
Consider the dt-SwS $S_{np}$ as in Definition~\ref{def:cons} and the corresponding safety specification $\safe(X_0,X_u)$, and a non path-complete graph $\G_{np}$. The graph-based barrier function $\big (\G_{np}, \{\B_v, v \in \V_{np}\} \big)$ with $\B_v$ as in~\eqref{eq:gbfnec}-\eqref{eq:sat} satisfies conditions~\eqref{eq:pcbf1}-\eqref{eq:pcbf2}. Moreover, condition~\eqref{eq:pcbf3} is satisfied if and only if for all $(v,\sigma,v') \in \E_{np}$:
\begin{equation} \label{eq:vecineq}
4\hat{A}_\sigma^\top p_{v'} \leq p_v.
\end{equation}
\end{proposition}
\begin{proof}
We first show the satisfaction of condition~\eqref{eq:pcbf1}. For $x \in X_0$ as in~\eqref{eq:init}, we have from~\eqref{eq:sat} that for all $v \in \V_{np}$, $\B_v(x) = \sum_{i=1}^n p_v[i]x[i]^2 - 1 \leq 4^k \sum_{i=1}^n x[i]^2 - 1 \leq 0$. Similarly, one can show the satisfaction of condition~\eqref{eq:pcbf2}. Now we show that condition~\eqref{eq:pcbf3} is satisfied if and only if the vector inequality~\eqref{eq:vecineq} holds. To do so, we first show the direction~\eqref{eq:pcbf3} $\implies$~\eqref{eq:vecineq}. By applying condition~\eqref{eq:pcbf3} to the dt-SwS $S_{np}$, we have that for all $(v,\sigma,v') \in \E_{np}$ and all $x \in \R^n$, $\B_{v'}(A_\sigma x) \leq \B_v(x)$. For an arbitrary index $0 \leq i' \leq n$, consider the $i'^{\text{th}}$ row of the inequality~\eqref{eq:vecineq}. If the $i'^{\text{th}}$ row of $A_\sigma$ is the $0$ vector, then we get $0 \leq p_v[i']$, which is automatically satisfied due to~\eqref{eq:sat}. If not, take any index $0 \leq i,i' \leq n$ such that $\hat A_\sigma[i',i] =1$ (note that for any $i$, there exists at most one $i'$ such that this holds) and fix $x=e_{i}$, which is the $i^{\text{th}}$ canonical basis vector.  By applying condition~\eqref{eq:pcbf3}, we get
\begin{equation*}
   4 p_{v'}[i'] \leq p_v[i]. 
\end{equation*}
Combining the inequalities and putting it in a matrix form, we get~\eqref{eq:vecineq}.

To show that~\eqref{eq:vecineq} $\implies$~\eqref{eq:pcbf3}, write for all $v \in \V_{np}$, $P_v = \diag(p_v)$. Then, we have
\begin{align*}
\B_{v'}(A_\sigma x) & = ((A_\sigma x)^\top P_{v'}(A_\sigma x)) - 1 \\ &= x^\top(2\hat A_\sigma )^\top P_{v'}(2 \hat A_\sigma) x - 1\\ & = 4 x^\top \hat A_\sigma^\top (\sum_{i'} p_{v'}[i']e_{i'} e_{i'}^\top) \hat A_\sigma x - 1 \\
& = 4x^\top \big(\sum_{i'} p_v[i'] (\hat A_\sigma^{\top} e_{i'})(\hat A_\sigma e_{i'}^\top) \big) x - 1\\
& \leq x^\top \sum_{i} p_v[i]e_{i} e_{i}^\top x -1 \\
& = x^\top P_v x - 1 = B_v(x),
\end{align*}
where the inequality arises from condition~\eqref{eq:vecineq}.$\qedsymb$
\end{proof}

Proposition~\ref{prop:vectorineq} allows us to show that the graph-based barrier function constructed in Section~\ref{subsec:cons} satisfies conditions~\eqref{eq:pcbf1}-\eqref{eq:pcbf2}. As the last step in constructing the proof of statement~\ref{ite:2nec} of Theorem~\ref{thm:nec}, we show that for the non path-complete graph $\G_{np}$, the construction also satisfies~\eqref{eq:pcbf3}, which can be shown by proving that the vector inequality~\eqref{eq:vecineq} holds. To do so, we utilize an auxiliary graph $\hat \G_{np} = (\hat \V_{np}, \hat \E_{np})$, whose edges describe the element-wise inequalities of~\eqref{eq:vecineq}. The nodes of $\hat \G_{np}$ are defined as
\[ \hat \V_{np} \coloneq \{(v,i) \mid v \in \V_{np}, 1 \leq i \leq n\}, \]
i.e., they are pairs of the nodes of $\G_{np}$ and the $i^{\text{th}}$ dimension of the dt-SwS $S_{np}$, such that each node represents a particular coefficient of the graph-based barrier function $p_v[i]$. Now, we construct the edge in $\hat \E_{np}$ from $(v,i)$ to $(v',i')$ if and only if
\begin{enumerate}
\item There is an edge $(v,\sigma,v') \in \V$ for some $\sigma \in \Sigma$, and
\item The corresponding matrix $A_\sigma$ is such that $A_\sigma (i',i) \neq 0$. 
\end{enumerate}
Correspondingly, we label the edge $\sigma$, i.e., $\big((v,i),\sigma,(v',i')\big)$ $\in \hat \E_{np}$. The graph $\hat \G_{np}$ enjoys some interesting properties, as shown in the following proposition, which we will leverage to show the satisfaction of~\eqref{eq:vecineq}. 

\begin{proposition} \label{prop:acyclic}
Given a non path-complete graph $\G_{np}$ over the alphabet $\Sigma$, the corresponding auxiliary graph $\hat \G_{np}$ is acyclic with at most $k-1$ consecutive edges.  
\end{proposition}

\begin{proof}
A cycle in the graph $\hat \G_{np}$ corresponds to a path $\bm{\pi}_l = \big(\big((v,i),\sigma_0,(v_1,i_i)\big)\ldots  \big((v_{l-1},i_{l-1}),\sigma_{l-1},$ $(v,i)\big)\big)$. Such a cycle can exist only if the $[i,i]$ element of the matrix product $A_{\bm{\sigma}_l} = A_{\sigma_{l-1}}\cdots A_{\sigma_0}$ is non-zero, where $\bm{\sigma}_l = (\sigma_0\ldots\sigma_{l-1})$. Moreover, we have that $A_{\bm{\sigma}_l}[i,i] =  \sum_{j_{l-2},\ldots,j_0=1}^{n} A_{\sigma_{l-1}}[i,j_{l-2}] \cdots A_{\sigma_0}[j_0,i]$.
By the construction of $A_\sigma$ by Definition~\ref{def:cons}, we have that $A_\sigma[0,j] = 0$ for  any $1 \leq j \leq n$ and any $\sigma \in \Sigma$, so, $A_{\bm \sigma_l}[0,0] = 0$. For any other $i$, again from Definition~\ref{def:cons}, we have that for any $\sigma \in \Sigma$, $A_\sigma[i,j] = 0$ whenever $ i \neq j+1$. Therefore, for $A_{\bm{\sigma}_l}[i,i] \neq 0$, one needs $i = j_{k-2} +1 = j_{k-3} +2 \cdots = i + (k-1)$, which is impossible. Therefore, $A_{\bm{\sigma}_l}[i,i] = 0$ for all $1 \leq i \leq n$, and thus, $\hat \G_{np}$ is acyclic. 

Now, if $\hat \G_{np}$ consists of $k$ consecutive edges, then it corresponds to a path $\bm{\pi}_k = \big(\big((v,i),\sigma_0,(v_1,i_1)\big)\ldots $ $\big((v_{k-1},i_{k-1}),\sigma_{k-1}, (v',i')\big)$. Then, from Lemma~\ref{lem:cons}, the product $A_{\bm{\sigma}_k} \neq 0$ if and only if $\bm{\sigma}_k$ exists in $\hat \G_{np}$, and correspondingly $\bm{\sigma}_k$ is accepted by $\G_{np}$. But this is not possible since $\G_{np}$ is non path-complete and cannot accept $\bm{\sigma}_k$. Therefore, $\hat \G_{np}$ can have at most $k-1$ edges.  
$\qedsymb$
\end{proof}

We now have all the prerequisites to show the existence of an admissible graph-based barrier function $(\G_{np}, \{\B_v, v \in \V_{np} \})$ as defined in~\eqref{eq:gbfnec} such that it satisfies conditions~\eqref{eq:pcbf1}-\eqref{eq:pcbf3}, for the constructed dt-SwS $S_{np}$ and the given non path-complete graph $\G_{np}$. 

\begin{refproof}{\bf{of Theorem~\ref{thm:nec}.}}
Statement~\ref{ite:1nec} of Theorem~\ref{thm:nec} follows directly from the construction of the dt-SwS $S_{np}$ and $\safe(X_0,X_u)$ in Definition~\ref{def:cons}, as $S_{np}$ violates $\safe(X_0,X_u)$ in $k$ time steps with $\bfm{x_{x_0}, \bm{\sigma}}(k) \in X_u$ for the sequence $\bm{\sigma_k}$ that is not accepted by $\G_{np}$. So we focus on the proof of statement (\ref{ite:2nec}). Note that every directed acyclic graph has a topological ordering of its nodes~\cite[Chapter 1]{digraphs}. This implies that the nodes of $\hat \V_{\np}$ can be re-numbered, i.e., $(v,i) \to s$, $1 \leq s \leq |\V_{np}|$, such that there exists a path from $s$ to $s'$ only if $s' > s$, $1 \leq s, s' \leq |\V_{np}|$. Furthermore, for each node $(v,i) \in \hat \V_{np}$ numbered $s$, assign the coefficient of the graph-based barrier function as
\begin{equation*}
p_v(i) = 4^{a_{s}}.
\end{equation*}
We determine $a_{s}$ as follows. If there exists an edge $\big((v,i),\sigma,(v',i')\big) \in \hat \E_{np}$ for some $\sigma \in \Sigma$, and the nodes $(v,i),(v',i')$ are numbered $s, s'$ respectively, then $a_{s} = a_{s'} + 1$. Otherwise, for all nodes $(v,i) \in \E_{np}$ with no outgoing edge, $a_{s} = 4$. Then, since the graph has no more than $k-1$ consecutive edges, one has that for all $(v,i) \in \E_{np}$, $4 \leq p_v(i) \leq 4^k$, satisfying condition~\eqref{eq:sat}. Therefore, using Proposition~\ref{prop:vectorineq}, we can show the satisfaction of conditions~\eqref{eq:pcbf1}-\eqref{eq:pcbf2}. 

Now, we have to show that for every edge $(v,\sigma,v') \in \E_{np}$ of $\G_{np}$, condition~\eqref{eq:pcbf3} holds. By Proposition~\ref{prop:vectorineq}, we instead show the satisfaction of condition~\eqref{eq:vecineq}. Consider the $i^{\text{th}}$ component of inequality~\eqref{eq:vecineq}. This corresponds to the term $4\hat A_\sigma^\top p_{v'}[i]$. If this is zero, then~\eqref{eq:vecineq} automatically holds at the $i^{\text{th}}$ component. In the case that $\hat A_\sigma^\top [i,i'] = 1$ for some index $i'$, this corresponds to an edge $\big((v,i),\sigma,(v',i')\big) \in \hat \E_{np}$ of $\hat \G_{np}$. As a result, we have that $4p_{v'}[i'] = 4^{a_{s'} + 1} \leq p_v[i]$. Therefore, condition~\eqref{eq:vecineq} is satisfied at all components $1 \leq i \leq n$. Thus, the graph-based barrier function $(\G_{np}, \{\B_v, v \in \V_{np} \})$ is indeed admissible, and the proof is complete. $\qedsymb$
\end{refproof}

In the following, we illustrate a simple counterexample obtained via the construction described above. 

\begin{example} \label{ex:nec}
Consider the non path-complete graph $\G_{np}=(\V_{np}, \E_{np})$ over $\Sigma=\{1,2\}$ as shown in Figure~\ref{fig:nonpc}. The corresponding finite sequence of length $k=3$ that cannot be accepted by $\G_{np}$ is given by $\bm{\sigma}_k = (121)$. Using this, we first construct a suitable dt-SwS $S_{np}$. By Definition~\ref{def:cons}, we obtain $\hat A_1 = \begin{bmatrix}
    0 & 0 & 0 & 0 \\
    1 & 0 & 0 & 0 \\
    0 & 0 & 0 & 0 \\
    0 & 0 & 1 & 0 \\
\end{bmatrix}$ and $\hat A_2 =
\begin{bmatrix}
    0 & 0 & 0 & 0 \\
    0 & 0 & 0 & 0 \\
    0 & 1 & 0 & 0 \\
    0 & 0 & 0 & 0 \\
\end{bmatrix}
$. Consequently, we have $A_1 = 2\hat A_1$, $A_2 = 2\hat A_2$, and $S_{np}$ given by~\eqref{eq:unsys}, with state set $X = \R^n$, $X_0 =  \{ x \in \R^n \mid \sum_{i=1}^{n} x[i]^2 \leq \frac{1}{64} \}$, $X_u =  \{x \in \R^n \mid \sum_{i=1}^n x[i]^2 \geq 1 \}$, and safety specification $\safe(X_0,X_u)$.

Clearly, statement~\ref{ite:1nec} of Theorem~\ref{thm:nec} is correct, since for the initial condition $x_0 = [\frac{1}{64} \ 0 \ 0 \ 0]^\top$ and switching sequence $\bm{\sigma} = (121\ldots)$, we have $\bfm{x}_{x_0, \bm{\sigma}}(3) = [ 0 \ 0 \ 0 \ 1] \in X_u$. Now, consider the graph-based barrier function $(\G_{np}, \{\B_v: v \in \{v_1,v_2\} \})$ of $\G_{np}$, where $\B_{v_1}(x) = 64x[1]^2 + 16 x[2]^2 
 + 16x[3]^2 + 4x[4] - 1$, and $\B_{v_2}(x) = 4x[1]^2 + 64 x[2]^2 + 4x[3]^2 + 4x[4]^2 - 1$. Note that these functions are obtained by utilizing the topological ordering of auxiliary graph $\hat \G_{np}$ corresponding to $\G_{np}$, shown in Figure~\ref{fig:Ghat}.
Indeed, the set of functions are admissible for $S_{np}$ and $\G_{np}$, since conditions~\eqref{eq:sat} and~\eqref{eq:vecineq} are satisfied. Therefore, $(\G_{np}, \{\B_v: v \in \{v_1,v_2\} \})$ cannot provide a valid set of safety certificates when $\G_{np}$ is not path-complete. 
\end{example}

\begin{figure}
  \centering
  \includegraphics[width=0.5\columnwidth]{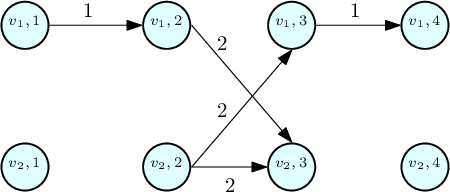}
  \caption{$\hat \G_{np}$ corresponding to graph $\G_{np}$ for Example~\ref{ex:nec}}
  \label{fig:Ghat}
\end{figure}

\section{Comparison of Path Complete Barrier Functions}
\label{sec:comparison}

In the previous section, we showed that graph-based barrier functions do not provide valid safety certificates unless they are path-complete. In general, given an alphabet $\Sigma$, the choice of a suitable path-complete graph for defining the safety conditions is not unique. As a result, the path-complete framework provides two degrees of freedom for the computation of suitable PCBFs. The first one is through the algebraic component, i.e., the template of barrier functions. For instance, one may change the template of the PCBF (e.g., from quadratic functions to general polynomial functions~\cite{prajna_worst_case,pushpak_cegis}, or even neural network parameterizations~\cite{zhao_learning_2021,anand_nn}) to improve conservatism of the safety certificates and making the safety verification problem less restrictive. Secondly, one can leverage the flexibility offered by the combinatorial component, i.e., the path-complete graph. While it was shown in~\cite{anand_pcbf} that path-complete barrier functions are generally less conservative than common and multiple barrier function frameworks~\cite{cbf_switch2,cbf_switch1,prajna_safety_2004}, it is not clear how the choice of different path-complete graphs may affect the conservatism.  In this section, we aim to address this problem by leveraging techniques from graph theory to compare two path-complete graphs and order them according to the conservatism offered by their corresponding safety conditions. We would like to note that the results here have been inspired from~\cite{geom_pclf,pclf_ordering1} which present approaches for systematically comparing the conservatism of Lyapunov functions and joint spectral radius for switched systems via path-complete Lyapunov functions.

Consider a dt-SwS $S$ defined as in~\eqref{eq:sysdyn}. For $S$,
suppose that $\G=(\V, \E)$ is a path-complete graph over $\Sigma$, and $(\G,\{\B_v, v \in \V \})$ is the corresponding PCBF. By a slight abuse of notation, we say that the PCBF $\B = (\G,\{\B_v, v \in \V \})$ belongs to the template $\mathcal{B}$, denoted by $\B \in \mathcal{B}$, if for all $v \in \V$, $\B_v \in \mathcal{B}$. For instance, when $\mathcal{B}$ is the set of quadratic functions, we want $\B_v$ to be a quadratic function for each $v$. Then, we define the ordering relation between any two path-complete graphs over $\Sigma$ as follows. 

\begin{definition}[Ordering of PCBFs]
\label{def:order}
Consider two path-complete graphs $\G$ and $\bar \G$. We say that $\G \preceq \bar \G$ if for any dt-SwS $S$ over $\Sigma$ as in~\eqref{eq:sysdyn}, any safety specification $\safe(X_0,X_u)$, and any template $\mathcal{B}$,  the existence of a suitable PCBF $\B = (\G, \{\B_v, v \in \V\}) \in \mathcal{B}$ implies the existence of a PCBF $\bar \B =(\bar \G, \{\bar \B_{\bar v}, \bar v \in \bar \V \}) \in \mathcal{B}$. Moreover, if the relation holds only for a particular template $\mathcal{B}$, we write $\G \preceq_{\mathcal{B}} \bar \G$. 
\end{definition}

One observes that the ordering relation $\G \preceq \bar \G$ means that the graph $\bar \G$ leads to potentially less conservative PCBFs in comparison with $\G$, since one may be able to find a PCBF corresponding to $\bar \G$ even when those corresponding to $\G$ do not exist. In the following, we focus on determining the ordering relation between two path-complete graphs $\G$ and $\bar \G$, by utilizing the notion of simulation relations, defined below. 

\begin{definition}[Simulation Relation]
\label{def:simu}
Consider two path-complete graphs $\G=(\V, \E)$ and $\bar \G=(\bar \V, \bar \E)$ defined over the same alphabet $\Sigma$. We say that $\G$ simulates $\bar \G$ if there exists a function $R: \bar \V \rightarrow \V$ such that
\begin{equation} \label{eq:simu}
(\bar v, \sigma, \bar v') \in \bar \E \implies \big(R(\bar v), \sigma, R(\bar v')\big) \in \E.
\end{equation}
\end{definition}
As an example, consider the path-complete graph $\bar \G$ in Figure~\ref{fig:pc_simu} and the path-complete graph $\G$ as shown in Figure~\ref{fig:pcgraph}. It is evident that $\G$ simulates $\bar \G$ with the relation $R(\bar v_1) = R(\bar v_3) = v_1$, and $R(\bar v_2) = v_2$. Now, we show the main result of this section, i.e., we prove that a simulation relation between graphs is both necessary and sufficient for determining the ordering relation between two path-complete graphs independent of the class of the systems and the considered template of the barrier functions.

\begin{figure}
    \centering
     \vspace{0.2em}
    \includegraphics[width=0.4\columnwidth]{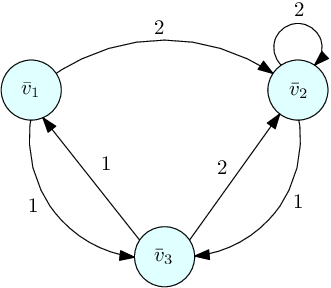}
    \caption{Path-complete graph $\bar \G$ that is simulated by the graph $\G$ in Figure~\ref{fig:pcgraph}.}
    \label{fig:pc_simu}
\end{figure}

\begin{theorem}[Ordering of PCBFs via Simulation]
\label{thm:simu}
Consider two path-complete graphs $\G = (\V, \E)$ and $\bar \G=(\bar \V, \bar \E)$ over the same alphabet $\Sigma$. Then the following statements are equivalent:
\begin{enumerate}
\item $\G$ simulates $\bar \G$ \label{ite:simu}
\item $\G \preceq \bar \G$ in the sense of Definition~\ref{def:order} \label{ite:ord}
\end{enumerate}
\end{theorem}

The proof of Theorem~\ref{thm:simu} is more involved. In fact, it relies on showing that for any given graph $\G = (\V, \E)$, adding another edge not already in the graph may lead to the violation of inequalities~\eqref{eq:pcbf1}-\eqref{eq:pcbf3}. This is formally stated as follows. 

\begin{lemma}
\label{lem:add}
For any graph $\G=(\V,\E)$ over an alphabet $\Sigma$, there exists a dt-SwS $S$ over $\Sigma$ as in~\eqref{eq:sysdyn}, a safety specification $\safe(X_0,X_u)$, and a suitable graph-based barrier function $(\G, \{\B_v, v \in \V\})$ such that the following hold:
\begin{alignat}{3}
& \forall (v,\sigma,v') \in \E, \forall x \in X: \quad & \B_{v'}(f_\sigma(x)) \leq \B_v(x), \label{eq:add1}  \\
& \forall (v,\sigma,v') \notin \E, \exists x \in X: \quad & \B_{v'}(f_\sigma(x)) > \B_v(x). \label{eq:add2}
\end{alignat}
\end{lemma}

\begin{proof}
The proof is by construction, i.e., for a given graph $\G = (\V, \E)$, we build a suitable switched system $S$ and a safety specification $\safe(X_0,X_u)$, and correspondingly also construct a graph-based barrier function $\{\G, \{\B_v, v \in \V \}\}$ that satisfies conditions~\eqref{eq:add1}-\eqref{eq:add2}. Consider $\tilde \E$ to consist of all edges not in $\E$, i.e., $\tilde \E = (\V \times \Sigma \times \V) \setminus \E$. For notational clarity throughout the proof, we distinguish edges in $\bar \E$ from those in $\E$ by using the symbol $e$ to denote the edges in $\bar \E$. Moreover, for each edge $e = (v, \sigma,v') \in \tilde \E$, we assign an index $1 \leq \tilde e \leq |\tilde \E|$. Now, we build a system $S$ of dimension $n$ and state set $X= \R^n$, where $n = 2|\tilde \E|$, as follows. For each $\sigma \in \Sigma$, consider $f_\sigma(x) = A_\sigma x$, where $A_\sigma$ is defined block-wise, with $|\tilde \E|$ blocks of $2 \times 2$ entries on the diagonal, and zero everywhere else. Then, for each $\tilde e \in \{1,\ldots,|\E|\}$, we set the $\tilde{e}^{\text{th}}$ diagonal block of $A_\sigma$ as $A_\sigma[\tilde e] = \begin{bmatrix} 
0 & 0 \\
1 & 0
\end{bmatrix}
$. Then, we consider the safety specification $\safe(X_0,X_u)$ with $X_0 \coloneq \{x \in \R^n \mid \sum_{i=1}^n x_i^2 \leq 1\}$, and $X_u = \{x \in \R^n \mid \sum_{i=1}^n x_i^2 \geq 3\}$. 

Now, we build an appropriate graph-based barrier function $\{\G, \{\B_v, v \in \V \}\}$ corresponding to $S$, $\safe(X_0,X_u)$, and $\G$. Suppose that $\B_v(x) = \sum_{i=1}^n q_v(i)x_i^2 - 1$, where $\frac{1}{3} \leq q_v(i) \leq 1$.  One can readily observe that conditions~\eqref{eq:pcbf1}-\eqref{eq:pcbf2} are satisfied, making $\{\G, \{\B_v, v \in \V \}\}$ a suitable graph-based barrier function as long as~\eqref{eq:pcbf3}, and correspondingly,~\eqref{eq:add1} is satisfied. Now, we need to assign coefficients $q_v$ that satisfy conditions~\eqref{eq:add1}-\eqref{eq:add2}. To do this, let $Q_v = \diag(q_v)$, then we can write $\B_v(x) = x^\top Q_v x$. The matrix $Q_v$ is then defined to be block-wise such that each diagonal block $Q_v[\tilde e]$ consists of coefficients $q_v[2\tilde e-1]$ and $q_v[2 \tilde e]$. With some abuse of notation, we rewrite each block diagonal matrix $Q_v[\tilde e]$ as a vector $Q_v[\tilde e] = \big[q_v[2\tilde e-1] \ q_v[2 \tilde e]\big]^\top$. Remark that for any $(v, \sigma, v') \in \E \cup \tilde \E$, the following holds:
\begin{align}
& \B_{v'}(A_\sigma x) \leq \B_v(x),\  \forall x \in \R^n, \iff 
 \forall e=(\tilde v, \sigma, \tilde v') \in \tilde \E: Q_{v'}^\top[\tilde e]\begin{bmatrix} 0 & 0 \\ 1 & 0 \end{bmatrix} \leq Q_{v}^\top[\tilde e]. \label{eq:equiv2}
\end{align}
The above relation is due to the construction of the dt-SwS $S$. Inequality~\eqref{eq:equiv2} means that the second element of $Q_{\tilde v'}[\tilde e]$ must be less than or equal to the first element of $Q_{\tilde v}[\tilde e]$. On the other hand, the second entry of $Q_v[\tilde e]$ can be arbitrary as long as it takes a value between $\frac{1}{3}$ and $1$, since the corresponding entry on the left hand-side is zero, and thus the lower part of inequality~\eqref{eq:equiv2} is automatically satisfied. 

Now, we are ready to formally obtain the elements $Q_v[\tilde e]$, for all $1 \leq \tilde e \leq |\tilde \E|$, which consequently define the coefficients of $\B_v(x)$. For each $v \in \V$ of $\G$, we decide the coefficients based on the indices $1 \leq \tilde e \leq |\tilde \E|$ and its corresponding edges $e = (\tilde v, \sigma, \tilde v') \in \tilde \E$, using the following cases:
\begin{enumerate}
\item $v = \tilde v = \tilde v'$ ($e$ is a self-loop from $v$): $Q_v[\tilde e] = [\frac{1}{3} \ \frac{1}{2}]^\top$, \label{case:self}
\item $v =\tilde v$ and $\tilde v \neq \tilde v'$ ($e$ is an outgoing edge from $v$): $Q_v[\tilde e] = [\frac{1}{3} \ \frac{1}{3}]^\top$, \label{case:out}
\item $v = \tilde v'$ and $\tilde v \neq \tilde v'$ ($e$ is an incoming edge to $v$):  $Q_v[\tilde e] = [\frac{1}{2} \ \frac{1}{2}]^\top$, \label{case:in}
\item otherwise:  $Q_v[\tilde e] = [\frac{1}{2} \ \frac{1}{3}]^\top$. \label{case:other}
\end{enumerate}

Now, we proceed by showing that the construction satisfies condition~\eqref{eq:add1} for any edge $(v, \sigma, v') \in \E$, by equivalently checking condition~\eqref{eq:equiv2} for all $e = (\tilde v, \sigma, \tilde v') \in \tilde \E$. The following scenarios are considered:
\begin{itemize}
\item Case 1: $v = v'$. We clearly cannot consider $v = \tilde v = \tilde v'$ since we get $e \in \E$, leading to a contradiction. In all other possible scenarios, the second term of $Q_{v}[\tilde e]$ is at least as small as the first term of $Q_v[\tilde e]$, so~\eqref{eq:add1} is satisfied.
\item Case 2: $v \neq v'$. This case can be further divided into the following scenarios, out of which the possible ones are summarized in Table~\ref{tab:scen}.
\begin{enumerate}[label={(2.\arabic*):}]
\item  $v = \tilde v \neq v' = \tilde v'$. Then, once again, we have $e \in \E$, so we can rule out this scenario. 
\item $v = \tilde v' \neq v' = \tilde v$. Then, $e$ is an incoming edge for the node $v$ and outgoing for the node $v'$. As a result, $Q_v[\tilde e] = [\frac{1}{2} \ \frac{1}{2}]^\top$, and $Q_{v'} = [\frac{1}{3} \ \frac{1}{3}]^\top$. 
\item $v = \tilde v' \neq \tilde v \neq v'$. Then, $e$ is an incoming edge to $v$ and for $\tilde v'$ it is neither. So $Q_v[\tilde e] = [\frac{1}{2} \ \frac{1}{2}]^\top$, and $Q_{v'}[\tilde e] = [\frac{1}{2} \ \frac{1}{3}]^\top$.
\item $v = \tilde v' = \tilde v \neq  v'$. Then, $e$ is a self-loop for $v$ and for $\tilde v'$ it is neither. So $Q_v[\tilde e] = [\frac{1}{3} \ \frac{1}{2}]^\top$, and $Q_{v'}[\tilde e] = [\frac{1}{2} \ \frac{1}{3}]^\top$.
\item $v = \tilde v \neq \tilde v' \neq v'$. Then, $e$ is an outgoing edge for $v$, and neither an incoming nor an outgoing one for $v'$. Therefore, 
$Q_v[\tilde e] = [\frac{1}{3} \ \frac{1}{3}]^\top$ and $Q_{v'}[\tilde e]= [\frac{1}{2} \ \frac{1}{3}]^\top$. 
\item $v \neq \tilde v \neq v' = \tilde v'$. Then $e$ is outgoing for $v'$ and neither for $v$. In this case, $Q_v[\tilde e] = [\frac{1}{2} \ \frac{1}{3}]^\top$ and $Q_{v'}[\tilde e] = [\frac{1}{3} \ \frac{1}{3}]^\top$.
\item $v \neq \tilde v = v' = \tilde v'$. Then $e$ is a self-loop for $v'$ and neither for $v$. In this case, $Q_v[\tilde e] = [\frac{1}{2} \ \frac{1}{3}]^\top$ and $Q_{v'}[\tilde e] = [\frac{1}{3} \ \frac{1}{2}]^\top$.

\item $v \neq \tilde v \neq v' = \tilde v'$. Then $e$ is incoming for $v'$ and neither for $v$. We have $Q_v[\tilde e] = [\frac{1}{2} \ \frac{1}{3}]^\top$, and $Q_{v'}[\tilde e] =  [\frac{1}{2} \ \frac{1}{2}]^\top$.
\item $v \neq \tilde v \neq v' \neq \tilde v'$. Then $e$ is neither an incoming nor an outgoing edge for $v, v'$. Then, we have $Q_v[\tilde e] = Q_{v'}[\tilde e] = [\frac{1}{2} \ \frac{1}{3}]^\top$.
\end{enumerate}
Clearly, in all of these cases, the second element of $Q_{v'}[\tilde e]$ is at least as small as the first element of $Q_v[\tilde e]$. Therefore, condition~\eqref{eq:equiv2}, and correspondingy condition~\eqref{eq:add1} follows.
\end{itemize}

\begin{table}[hbt!]
\centering
\caption{The eight possible scenarios considered when checking for condition~\eqref{eq:equiv2}.}
\begin{tabular}{|c|c|c|}
\hline
     $e = (\tilde v, \sigma, \tilde v') \in \tilde \E$ & $Q_v[\tilde e]$ & $Q_{v'}[\tilde e]$ \\
\hline
 $v = \tilde v' \neq v' = \tilde v $ & $[\frac{1}{2} \ \frac{1}{2}]^\top$ & $[\frac{1}{3} \ \frac{1}{3}]^\top$ \\
 
$v = \tilde v' \neq \tilde v \neq v'$ & $ [\frac{1}{2} \ \frac{1}{2}]^\top$ & $[\frac{1}{2} \ \frac{1}{3}]^\top$ \\

$v = \tilde v' \neq \tilde v \neq v'$ & $ [\frac{1}{3} \ \frac{1}{2}]^\top$ & $[\frac{1}{2} \ \frac{1}{3}]^\top$ \\

$v = \tilde v  \neq \tilde v' \neq v'$ & $ [\frac{1}{3} \ \frac{1}{3}]^\top$ & $ [\frac{1}{2} \ \frac{1}{3}]^\top$ \\

$v \neq \tilde v \neq v' = \tilde v'$ & $[\frac{1}{2} \ \frac{1}{3}]^\top$ & $[\frac{1}{3} \ \frac{1}{3}]^\top$ \\

$v \neq \tilde v = v' = \tilde v'$ & $[\frac{1}{2} \ \frac{1}{3}]^\top$ & $[\frac{1}{3} \ \frac{1}{3}]^\top$ \\

$v \neq \tilde v \neq v' = \tilde v'$ & $[\frac{1}{2} \ \frac{1}{3}]^\top$ & $[\frac{1}{2} \ \frac{1}{2}]^\top$ \\

$v \neq \tilde v \neq v' \neq \tilde v'$ & $[\frac{1}{2} \ \frac{1}{3}]^\top$ & $[\frac{1}{2} \ \frac{1}{3}]^\top$ \\
\hline    
\end{tabular}
\label{tab:scen}
\end{table}

Finally, we want to show the satisfaction of condition~\eqref{eq:add2}, by showing that for some $e=(\tilde v, \sigma, \tilde v')$,~\eqref{eq:equiv2} is violated. For any $(v,\sigma,v') \in \tilde \E$, consider the following cases by picking a suitable $e = (v, \sigma, v')$:
\begin{itemize}
\item Case 1: $v = v'$. Then, we have $Q_v[\tilde e] = Q_{v'}[\tilde e] = [\frac{1}{3} \ \frac{1}{2}]^\top$. The second element here is greater than the first one, so condition~\eqref{eq:equiv2}, and thus~\eqref{eq:add2} is violated.
\item Case 2: $v \neq v'$. In this case, $e$ is an outgoing edge for $v$ and an incoming one for $v'$. Then, $Q_v[\tilde e] = [\frac{1}{3} \ \frac{1}{3}]^\top$ and $Q_v'[\tilde e] = [\frac{1}{2} \ \frac{1}{2}]^\top$. Then, ~\eqref{eq:equiv2} and correspondingly,~\eqref{eq:add2} are violated. $\qedsymb$
\end{itemize}
\end{proof}

We now have the required ingredients to prove Theorem~\ref{thm:simu}.

\begin{refproof}{\bf{of Theorem~\ref{thm:simu}.}}
We first show that if statement~\eqref{ite:simu} of Theorem~\ref{thm:simu} holds, then statement~\eqref{ite:ord} holds as well. Suppose that $R: \bar \V \rightarrow \V$ is a simulation relation between $\G$ and $\tilde \G$. Moreover, suppose that for a given dt-SwS $S$ as in~\eqref{eq:sysdyn} and the graph $\G$, a suitable PCBF $\B = (\G, \{\B_v, v \in \V\})$ of some arbitrary template $\mathcal{B}$ exists. From this, one can easily construct a PCBF $\bar \B = (\bar \G, \{\bar \B_{\bar v}, \bar v \in \bar \V\})$ corresponding to $\bar \G$ by simply assigning to each node $\bar v \in \bar \V$, a corresponding function $\bar \B_{\bar v} = \B_{R(\bar v)}$. Note that such a function automatically satisfies conditions~\eqref{eq:pcbf1}-\eqref{eq:pcbf2}. Moreover, since $R$ is a simulation relation, for every edge $(\bar v, \sigma, \bar v') \in \bar \E$, one has that  $\bar \B_{\bar v'}(f_\sigma(x)) = \B_{R(\bar v')}(f_\sigma(x)) \leq \B_{R(\bar v)}(x) = \bar \B_{\bar v}(x)$, where the inequality is due to~\eqref{eq:simu} and the fact that $\{\B_v, v \in \V\}$ corresponds to the PCBF for $\G$ which satisfies condition~\eqref{eq:pcbf3}.
Therefore, one has the guarantee that existence of PCBF $\B = (\G, \{\B_v, v \in \V\}) \in \mathcal{B}$ implies the existence of PCBF $\bar \B = (\G, \{\bar \B_{\bar v}, \bar v \in \bar \V\}) \in \mathcal{B}$, and statement~\eqref{ite:ord} holds.

Then, we proceed to prove the converse statement, i.e., show that if statement~\eqref{ite:ord} holds, then~\eqref{ite:simu} also holds. Suppose that $\G \preceq \bar \G$, i.e., regardless of the template $\mathcal{B}$ and the system $S$, existence of PCBF $\B = (\G, \{\B_v, v \in \V\}) \in \mathcal{B}$ implies the existence of $\bar \B = (\G, \{\bar \B_{\bar v}, \bar v \in \bar \V\}) \in \mathcal{B}$. To show that $\G$ simulates $\bar \G$, we first apply Lemma~\ref{lem:add} to $\G$ and obtain a dt-SwS $S$ and a PCBF $\B = (\G, \{\B_v, v \in \V \} )$ such that conditions~\eqref{eq:add1}-\eqref{eq:add2} are satisfied.
Now choose the template $\mathcal{B} = \{B_v, v \in \V\}$. Since $\G \preceq \bar \G$, there must also exist a PCBF $\bar \B = (\G, \{\bar \B_{\bar v}, \bar v \in \bar \V\}) \in \mathcal{B}$. Now, define a function $R: \bar \V \to \V$ such that $\bar \B_{\bar v} = \B_{R(\bar v)}$, for all $\bar v \in \bar \V$. Note that this is always possible since the considered barrier function template $\mathcal{B}$ consists of the functions $B_v$, $v \in \V$. So we can write  
\begin{equation}
\hspace{-0.8em} \forall (\bar v, \sigma, \bar v') \in \bar \E: \B_{R(\bar v')}(f_\sigma(x)) \leq \B_{R(\bar v)}(x), \forall x \in X. \label{eq:simu_cont}
\end{equation}
We claim that such a function is a simulation relation, i.e., it satisfies condition~\eqref{eq:simu}. Suppose that it is not, i.e., for some $(\bar v, \sigma, \bar v') \in \bar \E$, we have that $\big(R(\bar v), \sigma, R(\bar v')\big) \notin \E$. Then, by Lemma~\ref{lem:add}, we know that there exists some $x \in X$ such that $\B_{R(\bar v')}(f_\sigma(x)) > \B_{R(\bar v)}(x)$. However, this is a contradiction with equation~\eqref{eq:simu_cont}. The proof is now complete. $\qedsymb$
\end{refproof}

\section{Case Studies and Experiments} \label{sec:cs}

\subsection{Nonlinear Switched Systems} \label{subsec:cs1}

This case study aims to demonstrate the applicability of the path-complete framework for the safety verification of nonlinear switched dynamical systems. Specifically, we consider a two-car platoon example that exhibits polynomial dynamics, enabling us to utilize the sum-of-squares (SOS) algorithm~\cite{Parrilo2003} to compute suitable polynomial PCBFs. Note that for the sake of brevity, we do not present the SOS algorithm in our paper. However, we refer the interested reader to some prior works~\cite{prajna_worst_case,pushpak_cegis, anand_compositional_2024}, and note that the computation of PCBF follows similarly.
Consider two cars moving in a platoon formation with the following discrete-time system dynamics, adapted from~\cite{luppi_data-driven_2024}:

\begin{align*}
x_1(t + 1) = &x_1(t) -\beta_1 x_1(t) - \alpha_1 x_1^2(t) + u_1, \\
x_2(t+1) = & x_2(t) -\beta_1 x_2(t) - \alpha_1 x_2^2(t) + u_2,
\end{align*}
where $x_1, x_2$ are the velocities of the vehicles $1$ and $2$, respectively, and $u_1, u_2$ are the forces normalized by the vehicle mass. Moreover, $\alpha_1 = 0.02, \alpha_2 = 0.04$ and $\beta_1 = 0.1, \beta_2 = 0.2$ are the rolling and aerodynamic-drag friction coefficients of the vehicles, respectively. Suppose that the second car, i.e., car corresponding to $x_2$ is the leader and communicates its state with the first car corresponding to $x_1$ at every time instant, and the corresponding inputs $u_1 = 0.01x_2, u_2 = 2$ are applied to the cars, respectively. However, occasionally and randomly, the communication between the cars breaks which results in $u_1 = 0$. This corresponds to a switched system as in equation~\eqref{eq:sysdyn}  with two operating modes, given by the dynamics
\begin{align*}
f_1(x) = &\begin{bmatrix}
0.01x_2 + 0.9x_1-0.02x_1^2, \\
2+0.8x_2-0.04x_2^2
\end{bmatrix}, \\
f_2(x) = &\begin{bmatrix}
0.9x_1-0.02x_1^2, \\
2+0.8x_2-0.04x_2^2
\end{bmatrix}.
\end{align*}
Furthermore, we define the safety specification by sets $X=[0,10]^{2}$, $X_0 = \{x_1 \in [0,3], x_2 \in X \mid 1 \leq x_2 - x_1 \leq 2\}$, and $X_u = \{x_1,x_2 \in X \mid x_2-x_1 \geq 0.2\}$. Now, we try to find a suitable PCBF by a priori choosing the path-complete graph given in Figure~\ref{fig:pcgraph_platoon} and a quadratic barrier function template. By casting the PCBF conditions as sum-of-squares constraints and utilizing the YALMIP optimization toolbox~\cite{Lofberg2004} with sedumi~\cite{sturm1999using} as the underlying solver, we are able to find a PCBF $(\G, \{\B_v, v \in V\}$ where $B_{v_1}(x) = 7.24+3.35x_1 - 8.91x_2+ 1.23x_1^2 + 0.91x_2^2 -0.76x_1x_2$ and $B_{v_2}(x) = 6.19+3.96x_1-8.47x_2+1.18x_1^2+0.86x_2^2-0.9x_1x_2$. The simulations shown in Figure~\ref{fig:platoon_simu1} also demonstrate the satisfaction 
of the safety specifications, thus illustrating our theoretical results. 
\begin{figure}
\centering
\includegraphics[scale=1.5]{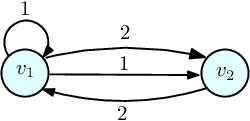}
\caption{PC graph $\mathcal{G}$ utilized for the examples in Section~\ref{subsec:cs1} and Section~\ref{subsec:cs_comp}.}  
\label{fig:pcgraph_platoon}
\end{figure}

\begin{figure}
\centering
\includegraphics[width=0.5\columnwidth]{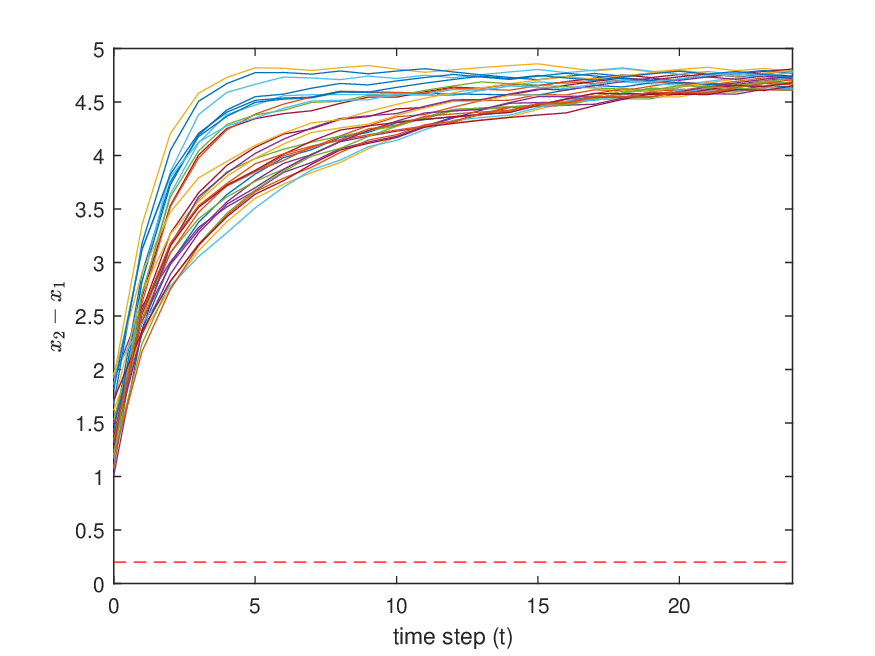}
\caption{Simulation corresponding to the platoon example in Subsection~\ref{subsec:cs1}. The value of $x_2-x_1$ is plotted against the time step $t$. The area below the red dotted line indicates the unsafe set.}
\label{fig:platoon_simu1}
\end{figure}

\subsection{Necessity of Path-Completeness}
\label{subsec:cs2}
Now for the two-car platoon example above, suppose that a control input of $u_1,u_2 = 0$ is applied when there is a communication breakdown, resulting in the modified second mode dynamics given by
\begin{equation*}
f_2(x) = \begin{bmatrix}
0.9x_1-0.02x_1^2, \\
0.8x_2-0.04x_2^2
\end{bmatrix},
\end{equation*}
while $f_1$ remains the same. We consider the same safety specification as considered above. Clearly, the simulations in Figure~\ref{fig:platoon_simu2} show that the switched system is unsafe. Now consider the non-path complete graph $\G_{np}$ obtained by removing the edge $(v_2,2,v_1)$ from the graph $\G$ in Figure~\ref{fig:pcgraph_platoon}. We attempt to find a quadratic graph-based barrier function $(\G_{np}, \{\B_v, v \in \V_{np})$ corresponding to the system using YALMIP. The obtained functions are given by $B_{v_1}(x) = 1.89+5.72x_1-6.45x_2+0.8x_1^2+0.86x_2^2-1.43x_1x_2$, and $B_{v_2}(x) = 3.93+4.94x_1-5.30x_2+0.51x_1^2+0.48x_2^2-0.82x_1x_2$. This means we risk obtaining invalid safety certificates even when the system is unsafe by utilizing non path-complete graphs instead of path-complete ones. 
\begin{figure}[t!]
\centering
\includegraphics[width=0.5\columnwidth]{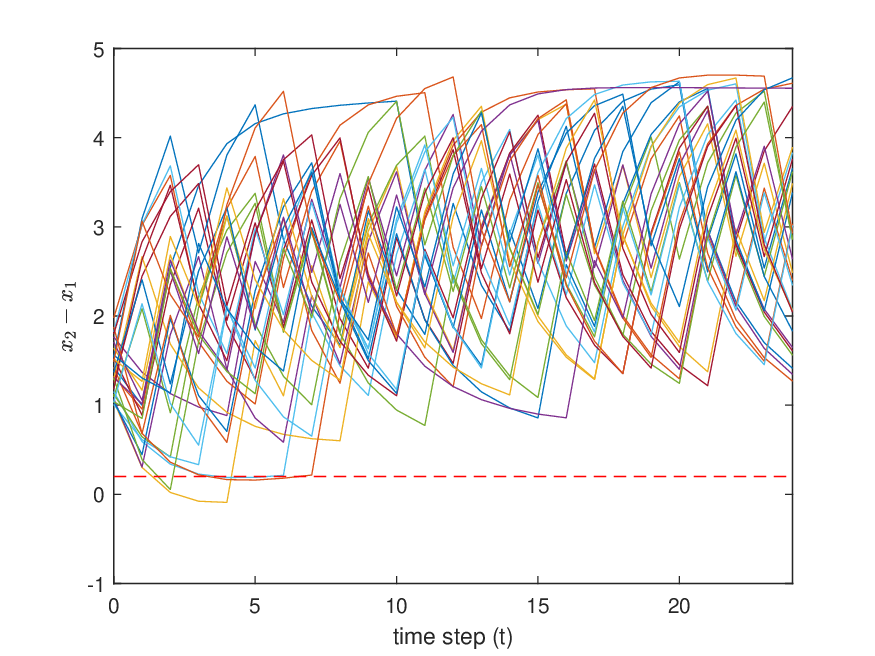}
\caption{Simulations corresponding to the platoon example in Subsection~\ref{subsec:cs2}. The simulations indicate the violation of safety specification.}
\label{fig:platoon_simu2}
\end{figure}

\subsection{Comparison of Path-Complete Barrier Functions}
\label{subsec:cs_comp}

To illustrate the comparison of two path-complete graphs that are related to each other via a simulation relation, we consider two path-complete graphs $\G$ and $\bar \G$, as shown in Figures~\ref{fig:pcgraph_platoon} and~\ref{fig:pcgraph_gbar}, respectively. Observe that $\G$ simulates $\bar \G$ via a simulation relation $R$ with $R(\bar v_1) = v_1, R(\bar v_2) = v_2$ and $R(\bar v_3) = v_1$. Now, we randomly sample $3000$ switched systems with two operating modes to examine the conservativeness of the PCBFs corresponding to $\G$ and $\bar \G$ respectively. To do so, we restrict ourselves to three-dimensional switched linear systems as in~\eqref{eq:sysdyn} given by $f_{\sigma}(x) = A_\sigma x$, $\sigma = \{1,2\}$. {We consider sets $X = \R^3$, $X_0 = \{x \in \R^n \mid \sum_{i=1}^3 x_i^2 \leq 4\}$ and $X_u = \{x \in \R^n \mid  \sum_{i=1}^3 x_i^2 \geq 9\}$. Moreover, we consider quadratic barrier functions in order to cast the PCBF conditions~\eqref{eq:pcbf1}-\eqref{eq:pcbf3} as linear matrix inequalities (LMIs)~\cite{anand_pcbf}. For each system and each mode $\sigma \in \{1,2\}$, we randomly generate $A_\sigma \in \R^{3 \times 3}$ such that its eigenvalues are less than $1$ to ensure stability. Specifically, the elements of the $3 \times 3$ matrices are obtained randomly from a uniform distribution. If the eigenvalues of are computed to be greater than $1$, the matrices are divided by a factor of $1.05$ until the eigenvalues are close to, but still less than, $1$. While stability in general is not required for safety verification, we restrict ourselves to stable systems to maximize the chances of finding PCBFs. 

Then, by solving the LMIs for each of the $3000$ systems we report that for $1435$ systems, we were not able to find suitable PCBFs corresponding to neither $\G$ nor $\bar \G$, and for $1511$ systems, we were able to find PCBFs corresponding to both graphs. However, there were $54$ cases for which we were able to find a PCBF corresponding to $\bar \G$ while the computation of PCBFs for $\G$ resulted in an infeasible optimization problem. On the other hand, there were no cases where we could find a PCBF corresponding to $\G$ when the computation for $\bar \G$ led to infeasibility. This shows that the graph $\bar \G$ tends to perform better in comparison to $\G$, and this is because $\bar \G$ consistently leads to less conservative PCBFs. This is a direct consequence of Theorem~\ref{thm:simu} because $\G$ simulates $\bar \G$. Note that the computations were performed using the YALMIP optimization toolbox~\cite{Lofberg2004} with sedumi~\cite{sturm1999using} as the underlying solver. 

\begin{figure}[t!]
\centering
\includegraphics[scale=1.2]{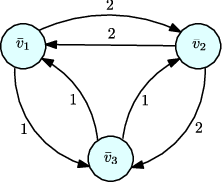}
\caption{PC graph $\bar \G$ utilized for the example in Section~\ref{subsec:cs_comp}.}  
\label{fig:pcgraph_gbar}
\end{figure}

\section{Conclusion and Discussion} 
\label{sec:conclusion}

In this paper, we addressed some important questions regarding the graph-based methodology for the safety verification of switched systems via path-complete barrier functions. We showed that it is necessary that a graph is path-complete to ensure the validity of the resulting safety certificates via barrier functions.  As a consequence, our framework provides a complete set of graph-based safety criteria for switched systems. Moreover, we proposed a systematic methodology for comparing two different path-complete graphs and the conservatism associated with the resulting safety conditions. In particular, we leveraged a notion of simulation relation between two graphs and showed that if this relation holds, one graph always performs better than the other and leads to less conservative results. This result is independent of the class of systems considered (i.e., linear, polynomial, etc.) as well as the algebraic templates of the barrier functions (quadratic, polynomial, etc.). Correspondingly, we demonstrated the soundness of our theoretical results with several case studies and experiments. 

Although the paper significantly advances the applicability of the path-complete framework for safety verification, there are still some gaps to be addressed. In particular, this paper proved that when two graphs are related by a simulation relation, one can characterize the relative conservatism of the resulting safety conditions irrespective of the considered system and the barrier function template. However, it is unclear how to compare the resulting conditions when two graphs are unrelated.
For example, consider the experiment performed in Section~\ref{subsec:cs_comp}. Suppose that we modify the graph $\bar \G$ in Figure~\ref{fig:pcgraph_gbar} slightly by reversing the edges $(\bar v_3,  1, \bar v_2)$ and $(\bar v_2, 2, \bar v_3)$. Then, we see that $\bar \G$ no longer simulates the graph $\G$ in Figure~\ref{fig:pcgraph_platoon} due to the presence of the edge $(\bar v_2,  1, \bar v_3)$ that can no longer be matched by $\G$ which has no outgoing edge labeled $1$ from $v_2$. However, we observed results similar to that of the experiment in this case, i.e., $\bar \G$ consistently provided less conservative results than $\G$. The reason for this is not yet clear. However, we postulate that there may be some specific system classes and barrier function templates for which one may characterize an ordering relation between graphs even when no simulation relation exists. Some results in this direction have been presented in the context of path-complete Lyapunov functions~\cite{pclf_ordering2, temp_lifts_pclf,jungers_statistical_2024}.

Finally, another interesting question is about the \emph{quality} of the computed safe set, i.e, the invariant set: while the stability problem is essentially a 'yes or no' answer, for safety analysis, the quality of the computed invariant set (e.g. its volume) may matter in practice. The relationship between the properties of a particular path-complete graph and the invariant set quantified by the corresponding PCBF is far from clear to us and would deserve further investigation.

\bibliographystyle{IEEEtran}
\bibliography{root.bib}

\end{document}